\documentclass[preprint,10pt]{elsarticle}
\usepackage{amssymb,amsmath}
\usepackage{amsthm}
\usepackage{graphicx}
\usepackage{color}
\usepackage{epstopdf}
\usepackage{subfigure}
\newcommand{\hide}[1]{}
\newcommand{\QED}{\hfill$\qed$}
\newtheorem{theorem}{Theorem}
\newtheorem{lemma}{Lemma}

  \newtheorem{example}{Example}
 \newtheorem{algorithm}{Algorithm}
 \newtheorem{procedure}{Procedure}
  
 \usepackage{xcolor}
\usepackage{lineno}

\usepackage{xcolor}

\journal{Discrete Applied Mathematics}
\begin{document}
%\linenumbers
\begin{frontmatter}

\title{Linear-Time Fitting of a $k$-Step Function}%
\author{Binay Bhattacharya %\fnref{label2}
}
\address{School of Computing Science, Simon Fraser University,
Burnaby, Canada}
\ead{binay@sfu.ca}
\author{Sandip Das}
\ead{sandip.das.69@gmail.com}
\address{Advanced Computing and Microelectronics Unit, Indian Stat. Inst., Kolkata, India}
\author{Tsunehiko Kameda}
\ead{tikokameda@gmail.com}
\address{School of Computing Science, Simon Fraser University,
Burnaby, Canada}

\begin{abstract}
Given a set of $n$ weighted points on the $x$-$y$ plane,
we want to find a step function consisting of $k$ horizontal steps
such that the maximum weighted vertical distance from any point to a step is minimized.
Using the prune-and-search technique,
we solve this problem in $O(n)$ time when $k$ is a constant.
Our approach can be applied directly or with small modifications
to solve other similar problems,
such as the maximum error histogram problem
and the line-constrained $k$-center problem,
in $O(n)$ time when $k$ is a constant.
\end{abstract}
\end{frontmatter}

%%%%%%%
\section{Introduction}\label{sec:intro}
Given an integer $k > 0$ and a set $P$ of $n$ weighted points in the plane,
our objective is to fit a $k$-step function to them so that the maximum weighted
vertical distance of the points to the step function is minimized.
We call this problem the {\em $k$-step function problem}.
It has applications in areas such as geographic information systems, 
digital image analysis, data mining, facility locations, and
data representation (histogram), etc.

In the unweighted case, if the points are presorted,
Fournier and Vigneron \cite{fournier2011} showed that
the problem can be solved in linear time using the results of 
\cite{frederickson1991a,frederickson1984,gabow1984}.
Later they showed that the weighted version of the problem can also be solved
in $O(n\log n)$ time~\cite{fournier2013},
using Megiddo's parametric search technique \cite{megiddo1983a}.
Prior to these results, the problem had been discussed by several researchers
\cite{chen2009,diazbanez2001,liu2010,lopez2008,wang2002}. 

Guha and Shim \cite{guha2007} considered this problem in the context of {\em histogram construction}.
In database research, it is known as the {\em maximum error histogram} problem.
In the weighted case,
 this problem is to partition the given points into $k$ buckets based on their $x$-coordinates,
such that the maximum $y$-spread in each bucket is minimized.
This problem is of interest to the data mining community as well (see \cite{guha2007} for references).
Guha and Shim \cite{guha2007} computed the optimum histogram of size $k$,
minimizing the maximum error.
They present algorithms which run in linear time when the points are unweighted,
and in $O(n\log n + k^2\log^6n)$ time and $O(n\log n)$ space when the points are weighted.

Our objective is to improve the above result to $O(n)$ time when $k$ is a constant.
We show that we can optimally fit a $k$-step function to unsorted weighted points in linear time.
We earlier suggested a possible approach to this problem at an OR workshop~\cite{bhattacharya2013b}.
Here we flesh it out, presenting a complete and rigorous algorithm and proofs.
Our algorithm exploits the well-known properties of prune-and-search along the lines in \cite{bhattacharya2007}.

This paper is organized as follows.
Section \ref{sec:prelim} introduces the notations used in the rest of this paper.
It also briefly discusses how the prune-and-search technique can be used
to optimally fit a $1$-step function (one horizontal line) to a given set of weighted points.
We then consider in Section~\ref{sec:cond2step} a variant of the 2-step function problem,
called the anchored 2-step function problem.
We discuss a ``big partition'' in the context of the $k$-partition of a point set
corresponding to a $k$-step function in Section \ref{sec:kstep}.
Section \ref{sec:algorithm} presents our algorithm for the optimal $k$-step function problem.
Section \ref{sec:conclusion} concludes the paper,
mentioning some applications of our results.

%%%%%%%
\section{Preliminaries}\label{sec:prelim}
%%%%%
\subsection{Model}\label{sec:model}
Let $P=\{p_1,p_2,\ldots, p_n\}$ be a set of $n$ weighted points in the plane.
For $1\leq i\leq n$ let $p_i.x$ (resp. $p_i.y$) denote the $x$-coordinate (resp. $y$-coordinate)
of point $p_i$, and let $w(p_i)$ denote its weight.
The points in $P$ are not sorted,
except that $p_1.x\leq p_i.x\leq p_n.x$ holds for any $i=1, \ldots, n$.\footnote{For the sake
of simplicity we assume that no two points have the same $x$ or $y$ coordinate.
But the results are valid if this assumption is removed.
}
Let $F_k(x)$ denote a generic $k$-step function,
whose $j^{th}$ segment (=step) is denoted by $s_j$.
For $1\leq j \leq k-1$, segment $s_j$ represents a half-open horizontal interval $[s_j^{(l)}, s_j^{(r)})$
between two points $s_j^{(l)}$ and $s_j^{(r)}$.
The last segment $s_k$ represents a closed horizontal interval $[s_k^{(l)}, s_k^{(r)}]$.
Note that $s_j^{(l)}.y= s_j^{(r)}.y$,
which we denote by $s_j.y$.
We assume that for any $k$-step function $F_k(x)$,
 segments $s_1$ and $s_k$ satisfy $s_1^{(l)}.x = p_1.x$ and $s_k^{(r)}.x = p_n.x$,
respectively.
Segment $s_j$ is said to {\em span} a set of points $Q\subseteq P$,
if $s_j^{(l)}.x \leq p.x <s_j^{(r)}.x$ holds for each $p\in Q$.
A $k$-step function $F_k(x)$ gives rise to a {\em $k$-partition} of $P$, 
${\cal P}= \{P_j \mid j = 1, 2, \ldots, k\}$,
 such that segment $s_i$ spans $P_i$.
 It satisfies the {\em contiguity condition} in the sense that for each partition $P_j$,
if $a.x \leq b.x$ for $a,b\in P_j$,
then every point $p$ with $a.x \leq p.x \leq b.x$ also belongs to $P_j$.
In the rest of this paper, we
consider only partitions that satisfy the contiguity condition.
Fig.~\ref{fig:4steps} shows an example of fitting a 4-step function $F_4(x)$.
\begin{figure}[ht]
\centering
\includegraphics[height=2.8cm]{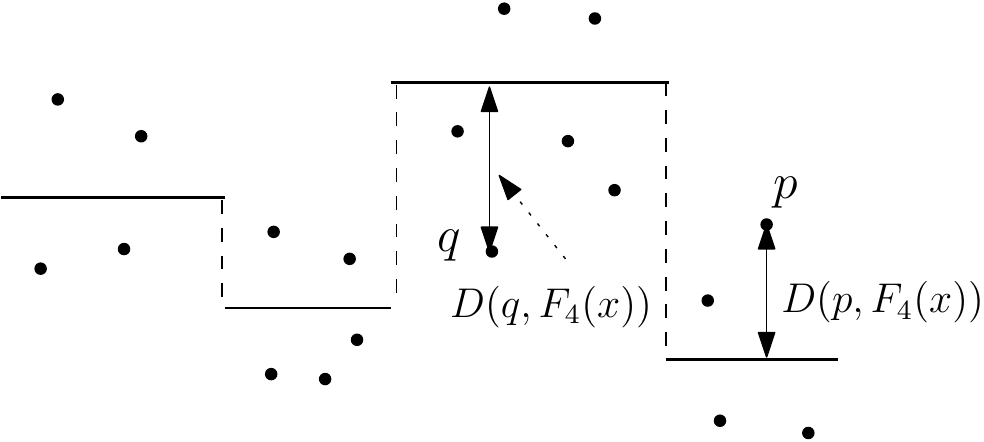}
\hspace{2mm}
\caption{Fitting a 4-step function.
}
\label{fig:4steps}
\end{figure}

Given a step function $F(x)$, 
defined over an $x$-range that contains $p.x$,
let $d(p, F(x))$ denote the vertical distance of $p$ from $F(x)$.
We define the {\em cost} of $p$ with respect to $F(x)$
by the weighted distance
\begin{equation}\label{eqn:pointCost}
D(p, F(x))\triangleq d(p, F(x))w(p).
\end{equation}
We generalize the cost definition for a set $Q \subseteq P$ of points by
\begin{equation}\label{eqn:setCost}
D(Q, F(x)) \triangleq \max_{p \in Q} \{D(p, F(x))\}.
\end{equation}
Point $p_h$ is said to be {\em critical} with respect to $F(x)$ if
\begin{equation}
D(p_h, F(x)) = D(P, F(x)).
\end{equation}
Note that there can be more than one critical point with respect to a given step function.
We similarly define a critical point with respect to a single segment of $F(x)$
as a point that has the maximum vertical weighted distance to it.

For a set of weighted points in the plane or on a line,
the point that minimizes the maximum weighted distance to
them is called the {\em weighted 1-center}~\cite{bhattacharya2007}.
Given a {\em $k$-partition} of $P$, 
${\cal P}= \{P_j \mid j = 1, 2, \ldots, k\}$,
it is clear that $\exists P_i \in {\cal P}$ such that $|P_i| \geq \lfloor n/k\rfloor$.
We call any such partition a {\em big partition}.
A big partition spanned by a segment in an optimal solution plays an important role.
(See Procedure {\tt Big$({\cal P}, k)$} in Sec.~\ref{sec:findBig}.)

%%%%
\subsection{Bisector}\label{sec:bisector}
If we map each point $p_i\in P$ onto the $y$-axis,
the {\em cost} of (or the weighted distance from) $p_i$
grows linearly from 0 at $p_i.y$ in each direction as a function of $y$.
Consider arbitrary two points $p$ and $q$.
Their costs intersect at either one or two points,
one of which always lies between $p.y$ and $q.y$.
If there are two intersections,
the other intersection lies outside interval $[p.y,q.y]$ on the $y$-axis.
If $p.y\not=q.y$ and $w(p)=w(q)$ hold
then there is only one intersection.\footnote{If $p.y=q.y$,
we can ignore one of the points with the smaller weight.
}
Let $a$ (resp. $b$) be the $y$-coordinate of the upper (resp. lower) intersection point,
where $b\leq a$.
We call the horizontal line $y=a$ (resp. $y=b$)
the {\em upper} (resp. {\em lower}) {\em bisector} of $p$ and $q$.
If there are only one intersection,
we pretend that there were two at $b=a$, which lies between $p.y$ and $q.y$.
(Note that the $y$-axis is shown horizontally in Figs.~\ref{fig:2points1} and \ref{fig:2points3} below,
where $y$ increases to the right.)

Let $\wp$ denote the $\lceil n/2\rceil$ pairs,
formed by pairing up the points in $P$ arbitrarily.
The cost lines of the points in each pair $(p,q)$ intersect
as shown in Figs.~\ref{fig:2points1}.
\begin{figure}[h]
\centering
\subfigure[]{\includegraphics[height=1.5cm]{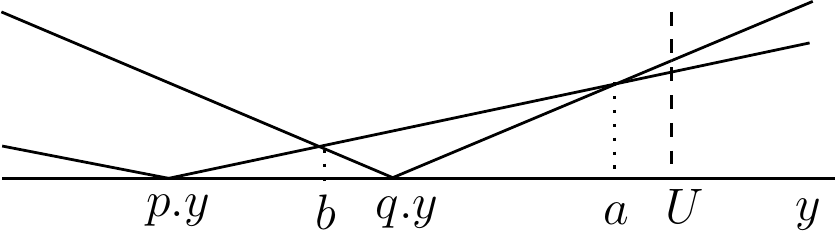}}
\hspace{2mm}
\subfigure[]{\includegraphics[height=1.5cm]{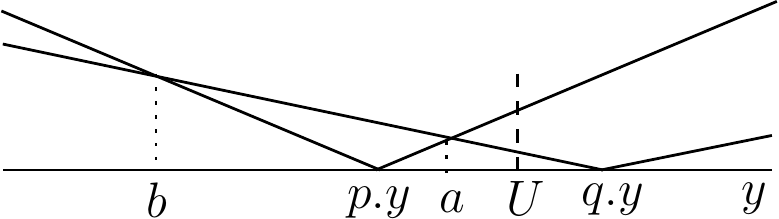}}
\caption{1/3 of upper intersections are at $y<U$:
(a) $p$ can be ignored at $y>U$; (b) $q$ can be ignored at $y>U$.
}
\label{fig:2points1}
\end{figure}
Let $y=U$ be the line at or above which at least 2/3 of the upper bisectors lie,
and at or below which at least 1/3 of the upper bisectors lie.
We use $\wp^U_{2/3}$ and $\wp^U_{1/3}$ to name the subsets of
 $\wp$ that have these two sets of bisectors, respectively.
 Note that $|\wp^U_{2/3}|\geq n/2\times 2/3 = n/3$ and $|\wp^U_{1/3}|\geq n/2\times 1/3 = n/6$.
Similarly, let  $y=L$ be the line at or below which at least 2/3 of the  
lower bisectors lie,
and  at or above which at least 1/3 of the bisectors lie.\footnote{{We define $U$ and $L$ this way,
because many points could lie on them.}}
We use $\wp^L_{2/3}$ and $\wp^L_{1/3}$ to name the subsets of
$\wp$ that have these two sets of bisectors, respectively.
 Note that $|\wp^L_{2/3}|\geq n/2\times 2/3 = n/3$ and $|\wp^L_{1/3}|\geq n/2\times 1/3 = n/6$.
\begin{lemma} \label{lem:one6th}
We can identify $n/6$ points that can be removed without affecting the weighted 1-center
 for the values
of their $y$-coordinates.
\end{lemma}
\begin{proof}
Consider the following three possibilities.
\begin{enumerate}
\item[(i)]
The weighted 1-center lies above $U$.
\item[(ii)]
The weighted 1-center lies below $L$.
\item[(iii)]
The weighted 1-center lies between $U$ and $L$,
including $U$ and $L$.
\end{enumerate}

In case (i), there are two subcases,
which are shown in Fig.~\ref{fig:2points1}(a) and (b), respectively.
Since the center lies above $U$, 
we are interested in the upper envelope of the costs in the 
$y$-region given by $y > U$.
In the case shown in Fig.~\ref{fig:2points1}(a),
the costs of points $p$ and $q$ satisfy $d(y,p.y)w(p) <  d(y,q.y)w(q)$ for $y > U$.
Thus we can ignore $p$.
In the case shown in Fig.~\ref{fig:2points1}(b),
 the costs of points $p$ and $q$ satisfy
$d(y,p.y)w(p) >  d(y,q.y)w(q)$ for $y > U$.
Thus we can ignore $q$.
Since $|\wp^U_{1/3}|\geq n/6$,
in either case, one point from each such pair can be ignored,
i.e., 1/6 of the points in $P$ can be eliminated, because it cannot affect the weighted 1-center.
In case (ii) a symmetric argument proves that 1/6 of the points in $P$ can be discarded.

\begin{figure}[ht]
\centering
\subfigure[]{\includegraphics[height=1.5cm]{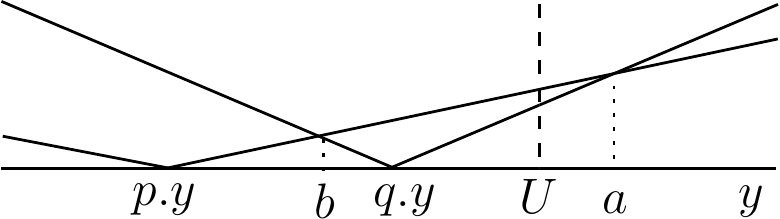}}
\hspace{2mm}
\subfigure[]{\includegraphics[height=1.5cm]{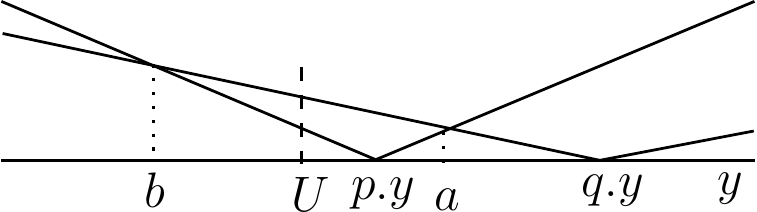}}
\caption{2/3 of upper bisectors are at $y> U$.}
\label{fig:2points3}
\end{figure}

In case (iii) see Fig.~\ref{fig:2points3}.
The costs of each pair in $\wp^U_{2/3}$ (of the $2n/3$ pairs) as functions of $y$ intersect at most once at $y<U$.
The cost functions of each pair in $\wp^L_{2/3}$ ($2n/3$ pairs) intersect at most once at $y>L$.
Therefore, $\wp^U_{2/3}\cap \wp^L_{2/3}$  ($n/3$ pairs must be common to both,)
i.e., both intersections of each such pair occur outside of the $y$-interval $[L,U]$.
$|\wp^U_{2/3}\cap \wp^L_{2/3}| = n/2 \times 1/3 = n/6$.
This implies that their cost functions do not intersect within in $[L,U]$,
i.e., one of each pair lies above that of the other in $[L,U]$,
and can be discarded.
\end{proof}

%%%%
\subsection{Optimal 1-step function}\label{sec:1step}
This problem is equivalent to finding the weighted center for $n$ points on a line.
We pretend that all the points had the same $x$-coordinate.
Then the problem becomes that of finding a weighted 1-center on a line,
i.e., on the $y$-axis.
This can be solved in linear time using Megiddo's {\em prune-and-search} method \cite{bhattacharya2007,chen2015a,megiddo1983a}.
In \cite{megiddo1983b} Megiddo presents a linear time algorithm in the case where the
points are unweighted. 
For the weighted case we now present a more technical algorithm that we can apply
later to solve other related problems.
The following algorithm uses a parameter $c$ which is a small integer constant.

\begin{algorithm}{\rm :} {\tt 1-Step}$(P)$
\begin{enumerate}
\item
Pair up the points of $P$ arbitrarily.
\item
For each such pair $(p,q)$ determine their horizontal bisector lines. 
\item
Determine a horizontal line, $y=U$ such that $|\wp^U_{2/3}|\geq  n/3$
and $|\wp^U_{1/3}|\geq n/6$ hold.
\item
Determine a horizontal line, $y=L$ such that $L$ $|\wp^L_{2/3}|\geq  n/3$ and
$|\wp^L_{1/3}|\geq n/6$ hold.
\item
Determine the critical points for $U$ and $L$.
\item
If there exist critical points for $U$ on both sides of (above and below) $U$, 
then $y=U$ defines an optimal 1-step function, $F^*_1(x)$; Stop. 
Otherwise, let $s_U$ (higher or lower than $U$) be the side of $U$ on which the critical point
lies.
\item
If there exist critical points for $L$ on both sides of $L$,
$y=L$ defines $F^*_1(x)$; Stop.
Otherwise, let $s_L$ (higher or lower than $L$) be the side of $L$ on which the critical point
lies.
\item
Based on $s_U$ and $s_L$, discard 1/6 of the points from $P$,
based on Lemma~\ref{lem:one6th}.
\item
If the size of the reduced set $P$ is greater than constant $c$, 
repeat this algorithm from the beginning with the reduced set $P$.
Otherwise, determine $F^*_1(x)$ using any known method
(which runs in constant time).
\end{enumerate}
\end{algorithm}

\begin{lemma}\label{lem:1step}
An optimal 1-step function $F^*_1(x)$ can be found in linear time.
\end{lemma}
\begin{proof}
The recurrence relation for the running time $T(n)$ of {\tt 1-Step}$(P)$ for general $n$ is 
$T(n) \leq T(n-n/6) + O(n)$,
which yields $T(n) = O(n)$.
\end{proof}

%%%%
\section{Anchored $2$-step function problem}\label{sec:cond2step}
In general, we denote an optimal $k$-step function by $F^*_k(x)$
and its $i^{th}$ segment by $s^*_i$. 
Later, we need to constrain the first and/or the last step of a step function to be
at a specified height.
A $k$-step function is said to be {\em left-anchored} (resp. {\em right-anchored}),
if $s_1.y$ (resp.  $s_k.y$) is assigned a specified value,
and is denoted by $^{\downarrow}\!F_k(x)$ (resp. $F_k^{\downarrow}(x)$).
The {\em anchored $k$-step function} problem is defined as follows.
Given a set $P$ of points and two $y$-values $a$ and $b$,
determine the optimal $k$-step function $^{\downarrow}\!F^*_k(x)$  (resp. $F_k^{\downarrow *}(x)$)
that is left-anchored (resp. right-anchored) at 
$a$ (resp. $b$) such that cost $D(P, ^{\downarrow}\!\!F^*_k(x))$ (resp. $D(P, F_k^{\downarrow *}(x))$)
 is the smallest possible.
If a $k$-step function is both left- and right-anchored, 
it is said to be {\em doubly anchored} and is
denoted by $^{\downarrow}\!F_k^{\downarrow}(x)$.

%%%%%
\subsection{Doubly anchored 2-step function}
Suppose that segment $s_1$ (resp. $s_2$) is anchored at $a$ (resp. $b$).
See Fig.~\ref{fig:anchored2}(a).
\begin{figure}[ht]
\centering
\subfigure[]{\includegraphics[height=3cm]{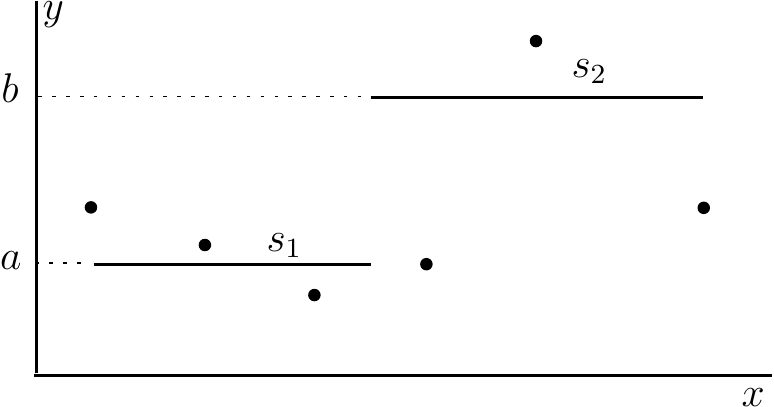}}
\hspace{4mm}
\subfigure[]{\includegraphics[height=3cm]{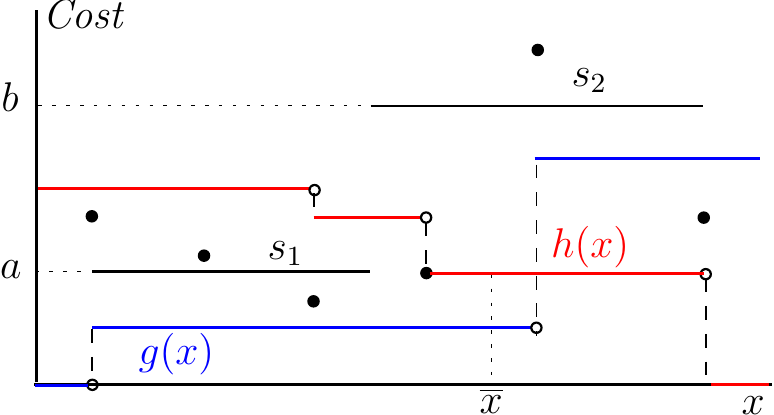}}
\caption{(a) $s_1.y=a$ and $s_2.y=b$;
(b) Monotone functions $g(x)$ (in blue) and $h(x)$ (in red).
}
\label{fig:anchored2}
\end{figure}
Let us define two functions $g(x)$ and $h(x)$ by
\begin{eqnarray}
g(x) &=& \max_{p.x\leq x} \{w(p)\cdot|p.y - a|~\mid p\in P\},\label{eqn:g}\\
h(x) &=&\max_{p.x>x} \{w(p)\cdot |p.y - b|~\mid p\in P\},\label{eqn:h}
\end{eqnarray}
where $g(x) =0$ for $x <p_1.x$ and $h(x) =0$ for $x >p_n.x$.
Intuitively, if we divide the points of $P$ at $x$ into two partitions $P_1$ and $P_2$,
then $g(x)$ (resp. $h(x)$) gives the cost of partition $P_1$ (resp. $P_2$).
See Fig.~\ref{fig:anchored2}(b).
Clearly the global cost for the entire $P$ is minimized for any $x$
at the lowest point in the upper envelope of $g(x)$ and $h(x)$,
which is named $\overline{x}$.
Since the points in $P$ are not sorted,
$g(x)$ and $h(x)$ are not available explicitly,
but we can compute $\overline{x}$ in linear time using the {\em prune-and-search} method,
taking advantage of the fact that $\max\{g(x),h(x)\}$ is unimodal.

\begin{algorithm}{\rm :} {\tt Doubly-Anch-2-Step}$(P,a,b)$\label{alg:double}
\begin{enumerate}
\item
Initialize $P'=P$.
\item
Find the point in $P'$ that has the median $x$-coordinate, $x_m$.
\item
Evaluate $g(x_m)$ (resp. $h(x_m)$) using (\ref{eqn:g}) (resp. (\ref{eqn:h})).
\item
If $g(x_m) = h(x_m)$ then $\overline{x}=x_m$. Stop.
\item
If $g(x_m) < h(x_m)$ (resp. $g(x_m) > h(x_m)$), 
i.e., $\overline{x}< x_m$ (resp. $\overline{x}< x_m$),
prune all the points $p$ with $p.x < x_m$ (resp. $p.x > x_m$),
 from $P'$,
remembering just the maximum cost.
\item
Stop when $|P'|=2$, and find the lowest point $\overline{x}$.
Otherwise, go to Step~2.
\end{enumerate}
\end{algorithm}

We have the following lemma.
\begin{lemma}\label{lem:doublyAnchored}
An optimal doubly anchored 2-step function
can be found in linear time.
\end{lemma}
\begin{proof}
Steps~2 and 3 of Algorithm {\tt Doubly-Anch-2-Step}$(P,a,b)$ can be carried out in linear time.
Since Step~4 cuts the size of $P'$ in half every time, Step~2 is entered $O(\log n)$ times.
Therefore the total time is $O(n)$.
\end{proof}

%%%%%%
\subsection{Left- or right-anchored 2-step function}
Without loss of generality, we discuss only a left-anchored 2-step function. 
Given an anchor value $a$,
we want to determine the optimal 2-step function with the constraint
that $s^*_1.y=a$, denoted by $^{\downarrow}\!F^*_2(x)$.
See Fig.~\ref{fig:anchored2}(a).
In this case, $b$ in (\ref{eqn:h}) is not given; 
we need to find the optimal value for it.
But assume for now that $b$ is also given,
and execute {\tt Doubly-Anch-2-Step}$(P,a,b)$.
From the solution that it yields,
can we find the direction in which to move $b$ to find the optimal
left-anchored 2-step function?
\begin{lemma}
Let $P_1$ (resp. $P_2$) be the left (resp right) partition of $P$ generated by {\tt Doubly-Anch-2-Step}$(P,a,b)$
such that  $s_1.y=a$ (resp. $s_2.y=b$), where $a<b$ without loss of generality.
Assume that $P_1$ is maximal in the sense that the boundary between $P_1$ and $P_2$
cannot be moved to the right without increasing the cost of the solution.
\begin{enumerate}
\item[(a)]
If $D(P_1, s_1) \geq D(P_2, s_2)$ then the optimal right-anchor cannot lie above $y=b$.
\item[(b)]
If $D(P_1, s_1) < D(P_2, s_2)$ and there is a critical point in $P_2$ for $s_2.y=b$ above $y=b$,
then the optimal right-anchor cannot lie below $y=b$.
\item[(c)]
If $D(P_1, s_1) < D(P_2, s_2)$ and there is a critical point in $P_2$ for $s_2.y=b$ below $y=b$,
then the optimal right-anchor cannot lie above $y=b$.
\end{enumerate}
\end{lemma}
\begin{proof}
(a) Assume first that the critical point $p_1$ for $s_1$ lies below $s_1$ ($y=a$).
Then we cannot reduce the cost by changing the value of $b$.
Therefore, assume that $p_1$ lies above $s_1$.
By the definition of $\{P_1,P_2\}$, the leftmost point in $P_2$ lies above $y=b$,
and moving the boundary between $P_1$ and $P_2$ to the right increases the cost,
which is due to the weighted distance between $s_1$ and the new point in $P_1$,
and this increase is independent of the value of $b$.
Moving this boundary to the left cannot decrease the cost,
until $p_1$ becomes a part of $P_2$,
and even then a decrease is not possible unless $b$ is made smaller.
Otherwise, $\{P_1,P_2\}$ wouldn't be optimal with the current $b$.

(b) We know that moving the boundary between $P_1$ and $P_2$ to the right increases the cost
 if $b$ is kept at the same value.
The cost increases if $b$ is made smaller.

(c) Symmetric to Case (b).
\end{proof}

\begin{figure}[ht]
\centering
\includegraphics[height=4cm]{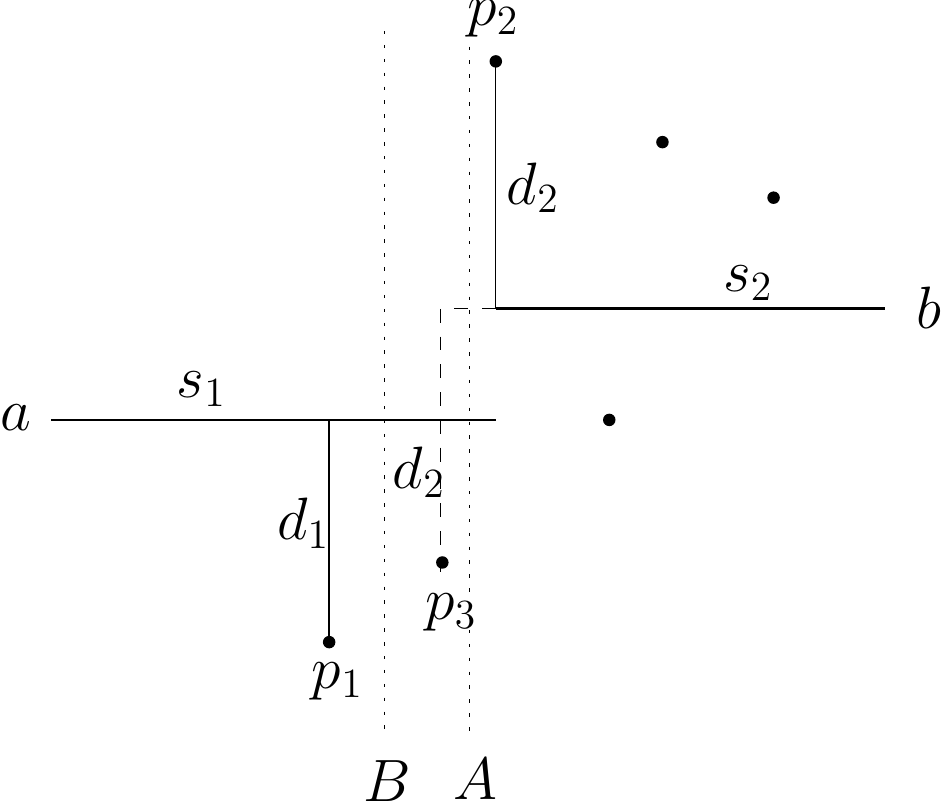}
\hspace{2mm}
\caption{An example for {\tt $l$-Anch-2-Step}$(P,a)$.
}
\label{fig:leftAnc}
\end{figure}
\begin{example}
In Fig.~\ref{fig:leftAnc},
assume that $d_2$ is slightly larger than $d_1$.
We have a doubly anchored solution with the minimum cost (weighted distance)
equal to $d_2$.
When the boundary between $P_1$ and $P_2$ is at $A$,
we can reduce the cost of the optimal solution by moving $b$ up.
We cannot do so if the boundary is at $B$,
because $D(p_3, s_2)$ would increase.
This is why we maximize $P_1$ in Step 5(b) of Algorithm~\ref{alg:cond2}
presented below.
\QED
\end{example}

To make use of the prune-and-search method,
we want to find the big partition (defined in Sec.~\ref{sec:model}),
$P_1$ or $P_2$, that is spanned by one segment of $^{\downarrow}\!F^*_2(x)$.
\hide{%%
\begin{procedure} {\tt Big2-$l$-Anc}$(P,a)$\label{proc:big2}
\begin{enumerate}
\item
Divide $P$ into left partition $P_1$ and right partition $P_2$,
whose sizes differ by at most one.\footnote{As before,
we assume that the points have different $y$-coordinates.
}
\item
Let $s_1$ be the segment with $s_1.y=a$ spanning $P_1$,
and let $s_2$ be the 1-step (optimal) solution for $P_2$.
\item
If $D(P_1, s_1) \leq D(P_2, s_2)$ (resp. $D(P_1, s_1) > D(P_2, s_2)$)
then $P_1$ (resp. $P_2$) is the big partition.
\end{enumerate}
\end{procedure}
}%end
If $P_1$ is the big partition, we can eliminate all the points belonging to it,
without affecting $^{\downarrow}\!F^*_2(x)$ that we will find.
See Step~4 of the Algorithm~\ref{alg:cond2} given below.
We then repeat the process with the reduced set $P$.
If $P_2$ is the big partition, on the other hand, we need to do more work,
similar to what we did to find an optimal 1-step function.
Namely,
we determine values $U$ and $L$ for $P_2$ by executing Algorithm {\tt 1-Step}$(P_2)$.
We then find a doubly anchored 2-step solution for $P$ with left anchor $a$
and right anchor $U$.

\begin{algorithm}{\rm :} {\tt $l$-Anch-2-Step}$(P,a)$\label{alg:cond2}
\begin{enumerate}
\item
Divide $P$ into left partition $P_1$ and right partition $P_2$,
whose sizes differ by at most one.\footnote{As before,
we assume that the points have different $y$-coordinates.
}
\item
Let $s_1$ be the segment with $s_1.y=a$ spanning $P_1$,
and let $s_2$ be the 1-step (optimal) solution for $P_2$.\footnote{
Segment $s_2$ can be found in $O(|P_2|)$ time by Lemma~\ref{lem:1step}.}
\item
If $D(P_1, s_1) = D(P_2, s_2)$ then
output $\{s_1,s_2\}$, which defines $^{\downarrow}\!F^*_2(x)$.
Stop.
\item
If $D(P_1, s_1) < D(P_2, s_2)$, remove from $P$ the points of $P_1$,
except the critical point for $s_1$.
Go to Step~6.
\item
If $D(P_1, s_1) > D(P_2, s_2)$ then carry out the following steps.
\begin{enumerate}
\item
Determine points $U$ and $L$ for $P_2$ as described in Algorithm {\tt 1-Step}$(P)$.
\item
Execute {\tt Doubly-Anch-2-Step}$(P,a,U)$,
and find the solution whose left partition is maximal.
Repeat it with right anchor $L$.
\item
Eliminate 1/6 of the points of $P_2$ from $P$, based on the two solutions
(as in Steps 6--8 of Algorithm {\tt 1-Step}$(P)$.)
\end{enumerate}
\item
If $|P| > c$ (a small constant),
repeat Steps~1 to 4.
Otherwise, optimally solve the problem in constant time, using a known method.
\end{enumerate}
\end{algorithm}
In the example in Fig.~\ref{fig:leftAnc}, 
assume that $b$ is not given,
and $s_2$ is determined by Step~2.
Then we have $D(P_1, s_1) > D(P_2, s_2)$,
and Step~5 applies.
According to Step~5(a), we determine $U$.
We then find the doubly anchored solution with the right anchor set to $b=U$.

\begin{lemma}
Algorithm {\tt $l$-Anch-2-Step}$(P,a)$ computes $^{\downarrow}\!F^*_2(x)$ correctly,
and runs in linear time.
\end{lemma}
\begin{proof}
Step~3 is obviously correct.
If $D(P_1, s_1) < D(P_2, s_2)$ holds in Step~4,
then the first partition of $^{\downarrow}\!F^*_2(x)$ contains $P_1$.
We need to keep the critical point for $a$,
but all other points of $P_1$ can be ignored from now on
because $P_1$ will expand.
If $D(P_1, s_1) > D(P_2, s_2)$ holds in Step~5,
then the first partition of $^{\downarrow}\!F^*_2(x)$ is contained in $P_1$.

Each iteration of Steps 3 and 4 will eliminate at least $1/2\times1/6=1/12$ of the points of $P$.
Such an iteration takes linear time in the input size. 
The total time needed for all the iterations is therefore linear.
\end{proof}

%%%%
\section{$k$-step function}\label{sec:kstep}
%%%%%%
\subsection{Approach}
To design a recursive algorithm, assume that for any set of points $Q\subset P$,
we can find the optimal $(j-1)$-step function and the optimal left- and right-anchored $j$-step function 
for any $2\leq j < k$ in $O(|Q|)$ time,
where $k$ is a constant .
We have shown that this is true for $k=2$ in the previous two sections.
So the basis of induction holds.

Given an optimal $k$-step function $F^*_k(x)$, for each $i~(1\leq i \leq k)$,
let $P^*_i$ be the set of points vertically closest to segment $s^*_i$.
By definition, the partition
$\{P^*_i \mid i = 1, 2, \ldots, k\}$ satisfies the contiguity condition.
It is easy to see that
for each segment $s^*_i$, there are (local) critical points with respect to $s^*_i$,
lying on the opposite sides of $s^*_i$.

In finding an optimal $k$-step function,
we first identify a big partition that will be spanned by a segment in
an optimal solution.
By Lemma~\ref{lem:big},
such a big partition always exists.
Our objective is to eliminate a constant fraction of the points in a big partition.
This will guarantee that a constant fraction of the input set is eliminated when $k$ is a fixed constant.
The points in the big partition other than two critical points are ``useless''
and can be eliminated from further considerations.\footnote{Note that there may
be more than two critical points in which case all but two are ``useless.''}
This elimination process is repeated until the problem size gets small enough
to be solved by an exhaustive method in constant time.

%%%%%%
\subsection{Feasibility test}\label{sec:feasibility}
Given a weighted distance (=cost) $D$,
a point set $P$ is said to be $D$-{\em feasible} if there exists a $k$-step function
$F_k(x)$ such that $D(P,F_k(x)) \leq D$. 
To test $D$-feasibility
we first try to identify the first segment $s_1$ of a possible $k$-step function $F_k(x)$.
To this end we compute the median $m$ of $\{p_i.x\mid i = 1, 2,\ldots, n\}$ in $O(n)$ time,
and divide $P$ into two parts $P_1 = \{p_i \mid p_i.x \leq m\}$ and $P_2 = \{p_i \mid p_i.x > m\}$,
which also takes $O(n)$ time.
Note that $|P_1| \leq \lceil |P|/2\rceil$ and $|P_2| \leq \lceil |P|/2\rceil$ hold.
We then find the intersection $I$ of the $y$-intervals in $\{|p_i.y-y| \leq D \mid p_i \in P_1\}$.
Assuming that $P$ is $D$-feasible,
then we have two cases.

Case (a): [$|I|=\emptyset$] $s_1$ ends at some point $p_j \in P_1$.
Throw away all the points in $P_2$ and look for the longest $s_1$ limited by cost $D$,
considering only the points in $P_1$ from the left.

Case (b): [$|I|\not=\emptyset$] $s_1$ may end at some point $p_j \in P_2$.
Throw away all the points in $P_1$
and look for the longest $s_1$, using $I$ and the points in $P_2$ from the left.
%After computing the intersection $I'$ of the $y$-intervals for the left half of $P_2$,
%$I$ should be updated to $I\cap I'$.

%Repeating this, we can find in $O(n)$ time the longest first step $s_1$ of a possible $F_k(x)$
%and the set of points that are at no more than distance $D$ from $s_1$.
Clearly,
we can find the longest $s_1$ in $O(n)$ time.
Remove the points spanned by $s_1$ from $P$,
and find $s_2$ in $O(n)$ time, and so on.
Since we are done after finding $k$ steps $\{s_1, \ldots, s_k\}$,
it takes $O(kn)$ time.

\begin{lemma}\label{lem:feasibility}
We can test $D$-feasibility in $O(kn)$ time.
\QED
\end{lemma}

%%%%%%
\subsection{Identifying a big partition}\label{sec:findBig}
\begin{lemma}\label{lem:big}
Let ${\cal P}=\{P_i\mid i=1,\ldots,k\}$ be any $k$-partition of $P$,
satisfying the contiguity condition,
such that the sizes of the partitions differ by no more than 1,
and let $\{P^*_i\mid i=1,\ldots,k\}$ be an optimal $k$-partition.
Then there exists an index $j$ such that $P_j$ is a big partition spanned by $s^*_j$.
\end{lemma}
\begin{proof} 
Let $j$ be the smallest index such that $s^{(r)}_j.x \leq s^{*(r)}_j.x$.
Such an index must exists, because if $s^{(r)}_j.x > s^{*(r)}_j.x$
for all $1\leq j \leq k-1$ then $s^{(r)}_k.x = s^{*(r)}_j.x$.
We clearly have $s_j\subset s^*_j$,
which implies that $s^*_j$ spans $P_j$.
\end{proof}

Given a point set $P$ in the $x$-$y$ plane,
let ${\cal P}=\{P_i\mid i=1,\ldots,k\}$ be any $k$-partition of $P$,
satisfying the contiguity condition,
such that the sizes of the partitions differ by no more than 1.
The following procedure returns a big partition $P_j$ spanned by $s^*_j$,
whose existence was proved by Lemma~\ref{lem:big}.
Since $P=\cup\{P_i \mid P_i \in {\cal P}\}$,
$P$ is implicit in the input to the next procedure.

\begin{procedure}{\rm :} {\tt Big$({\cal P}, k)$}\label{proc:bigk}

\noindent
%{\bf Output:} A big partition $P_j$ spanned by $s^*_j$ for some $j$.
\begin{enumerate}
\item
Using Algorithm~{\tt 1-Step}$(P)$, compute the optimal 1-step function for $P_1$
and let $D_1$ be its cost for $P_1$.
If $P$ is not $D_1$-feasible (i.e., $D(P,F^*_{k}(x))>D_1$),
then return $P_1$ and stop.\footnote{There exists an optimal solution for $P$
 in which $s^*_1$ spans $P_1$.}
\item
Using Algorithm~{\tt 1-Step}$(P)$, compute the optimal 1-step function for $P_k$
and let $D'_k$ be its cost for $P_k$.
If $P$ is not $D'_k$-feasible (i.e., $D(P,F^*_{k}(x))>D'_k$),
then return $P_k$ and stop.
\item
Find an index $j~ (1 < j < k)$ such that for $D_{j-1}=D(\cup_{i=1}^{j-1} P_i, F^*_{j-1}(x))$
$P$ is $D_{j-1}$-feasible, 
and for $D_j=D(\cup_{i=1}^{j} P_i, F^*_j(x))$ $P$ is not $D_{j}$-feasible.\footnote{This means
that $D_{j-1}\geq D^*$ and $D_{j}< D^*$,
where $D^*$ is the cost of the optimal solution for $P$.
Unless $P^*_i =P_i$ for all $i$, such an index $j$ always exists.
[We should indicate why.]}
Return $P_j$ and stop.
\end{enumerate}
\end{procedure}

\begin{lemma}\label{lem:Bigiscorrect}
Procedure {\tt Big$({\cal P},k)$} is correct.
\end{lemma}
\begin{proof}
It is clear that Steps~1 and 2 are correct.
To show that Step~3 is also correct, 
we {\em stretch} a step $s$ of an optimal step function
by making it as long as possible as follows.
Move $s^{(l)}.x$ (resp. $s^{(r)}.x$) to the left (resp. right) as far as possible without changing the cost
of the step function.
The step that has been stretched is called a {\em stretched step.} 
Let us assume without loss of generality that $s^*_j$ corresponding to $P_j$ returned by Step~3 is stretched.
Since $D_{j-1}\geq D^*$,
we must have $s^{*(l)}_j.x \leq s^{(l)}_j.x$.

The optimal solution $F^*_j(x)$ for $\cup_{i=1}^j P_i$ has cost $D_j$,
which is too small for $P$ to be $D_j$-feasible.
Regarding the remaining points $\cup_{i=j+1}^k P_i$,
let $G^*_j(x)$ denote the optimal $(k-j)$-step function for this point set.
If $D(\cup_{i=j+1}^k P_i, G^*_j(x)) \leq D_j$,
the $P$ would be $D_j$-feasible.
Since it is not,
 $s^{(r)}_j.x$ would be stretched to the right under the optimal solution $F^*_k(x)$,
i.e., $s^{*(r)}_j.x \geq s^{(r)}_j.x$.
Together with $s^{*(l)}_j.x \leq s^{(l)}_j.x$,
it follows that $P_j$ is spanned by $s^*_j$.
\end{proof}

\begin{lemma}\label{lem:Bigislinear}
Procedure {\tt Big$({\cal P},k)$} runs in linear time in $n$.
\end{lemma}
\begin{proof}
In Step~1, the optimal 1-step function for $P_1$ can be found in $O(|P_1|)$ time by Lemma~\ref{lem:1step},
and it takes $O(kn)$ time to test if $P$ is not $D_1$-feasible by Lemma~\ref{lem:feasibility}.
Similarly, Step~2 can be carried out in $O(n)$ time.
To carry out Step~3,
we compute, using binary search, $\lceil \log n\rceil$ values out of $\{D_i\mid 1\leq i \leq k-1\}$,
which takes $O(f(k)n)$ time for some function $f(k)$,
under the assumption that any $i$-step function problem, $i < k$,
is solvable in time linear in the size of the input point set.
\end{proof}

%%%%%
\section{Algorithm}\label{sec:algorithm}
\subsection{Optimal $k$-step function}
In this section we are assuming that we can solve any $(j-1)$-step and anchored
$j$-step function problems for any $2\le j <k$.
We have shown that this is true for $k=2$ in the previous section.
So the basis of recursion holds. 

Let us find an optimal doubly anchored $k$-step function, $^{\downarrow}\!F_k^{\downarrow *}(x)$,
which consists of $k$ horizontal segments,
$s_i^*, i=1,2, \ldots, k$,
satisfying $s_1^{*(l)}.x=p_1.x$,  $s_1^*.y=a$,   $s_k^{*(r)}.x = p_n.x$,  and $s_k.y=b$,
where $a$ and $b$ are given constants.
Let $P_i^*$  {be}  the set of points of $P$ vertically closest to $s_i^*$. For each segment $s_i^*$, there are critical points with respect to $s_i^*$, lying on the opposite sides of $s_i^*$.
In order to find $^{\downarrow}\!F_k^{\downarrow *}(x)$, 
we first execute {\tt Big$({\cal P}, k)$}
and identify a big partition containing at least $\lfloor n/k\rfloor$ points,
which are vertically closest to the same segment in some optimal solution.

Once a big partition, say $P_j$, is identified,
We first determine $U$ and $L$ for $P_j$ as described in Algorithm {\tt 1-Step}$(P)$.
To illustrate the idea,
let us consider a special case where $j=1$ and $k=2$.
We execute {\tt $l$-Anch-2-Step}$(P,U)$ and  {\tt $l$-Anch-2-Step}$(P,L)$
and determine $s_U$ and $s_L$ as in  Algorithm {\tt 1-Step}$(P)$.
We can thus eliminated 1/6 of the points in $P_1$.
We repeat this with the reduced $P$.
It may turn out that the right partition is the big partition in the next round.
Then we can repeat the above process symmetrically.
Eventually, the size of $P$ gets small enough,
so that we can find the solution using an exhaustive method.

For a general $k$ and $j>1$,
we need to find the left- and right-anchored solution for $U$ and $L$,
and prune $1/6$ of the points in $P_j$ using {\tt Prune-Big$(k,P_j)$}, given below,
which is very similar to Algorithm {\tt 1-Step}$(P)$.
Let $P_j$ be a big partition spanned by $s^*_j$,
which is an input to the following procedure.
\begin{procedure}{\rm :} {\tt Prune-Big$(k,P_j)$}

\noindent
{\bf Output:} 1/6 of points in $P_j$ removed.

\begin{enumerate}
\item
Determine $U$ and $L$ for $P_j$ as in Algorithm {\tt 1-Step}$(P)$. 
\item 
If $j>1$, find two right-anchored $j$-step functions $F_j^{\downarrow *}(x)$ for $\cup_{i=1}^{j} P_i$,
one anchored by $L$ and the other anchored by $U$.
\item
If $j< k$, find two left-anchored $(k-j+1)$-step functions
 $^{\downarrow}\!F^*_{k-j+1}(x)$ for $\cup_{i=j}^k P_i$,
one anchored by $L$ and the other anchored by $U$.
\item
Identify 1/6 of the points in $P_j$ with respect to $L$ and $U$,
which are ``useless''\footnote{See Step~8 of Algorithm~{\tt 1-Step}$(P)$.}
based on $F_j^{\downarrow *}\!(x)$
and $^{\downarrow}\!F^*_{k-j+1}(x)$ found above,
and remove them from $P$. 
\end{enumerate}
\end{procedure}

\begin{lemma}\label{lem:anchored}
{\tt Prune-Big$(k,P_j)$} runs in linear time when $k$ is a constant..
\QED
\end{lemma}

We can now describe our algorithm formally as follows.

\begin{algorithm}{\rm :} {\tt $k$-Step}$(P)$. 

\noindent
{\bf Output:}  Optimal $k$-step function $F^*_k(x)$
\begin{enumerate}
\item 
Divide $P$ into partitions $\{P_i \mid i = 1, 2, \ldots, k\}$,
satisfying the contiguous condition,
such that their sizes differ by no more than one.
\item 
Execute Procedure {\tt Big$({\cal P},k)$} to find a big partition $P_j$
spanned by $s^*_j$. 
\item
Execute Procedure {\tt Prune-Big$(k,P_j)$}.
\item
If $|P| > c$ for some fixed $c$, 
repeat Steps~1 to 3 with the reduced $P$.
\end{enumerate}
\end{algorithm}

%%%%%
\subsection{Analysis of algorithm}
To carry out Step 1 of Algorithm {\tt $k$-Step}$(P)$, 
we first find the $(hn/k)^{th}$ smallest among $\{p_i.x \mid 1\leq i \leq n\}$,
for $h=1, 2, \ldots, k-1$.
We then place each point in $P$ into $k$ partitions delineated by these
$k-1$ values.
It is clear that this can be done in $O(kn)$ time.\footnote{This could be done in $O(n\log k)$ time.}
As for Step~2,
we showed in Sec.~\ref{sec:findBig} that finding a big partition spanned by an optimal step
$s_j^*$ takes $O(n)$ time, since $k$ is a constant.
Step~3 also runs in $O(n)$ time by Lemma~\ref{lem:anchored}.
Since Steps 1 to 3 are repeated $O(\log n)$ times,
each time with a point set whose size is at most a constant fraction of the size of the previous set,
the total time is also $O(n)$, when $k$ is a constant.
By solving a recurrence relation for the running time of Algorithm~{\tt $k$-Step}$(P)$,
we can show that it runs in $O(2^{2k\log k}n)=O(k^{2k}n)$ time.

\begin{theorem}
Given a set of $n$ points in the plane $P=\{p_1,p_2,\ldots, p_n\}$,
we can find the optimal $k$-step function that minimizes the maximum distance
to the $n$ points in $O(k^{2k} n)$ time.
\QED
\end{theorem}
Thus the algorithm is optimal for a fixed $k$.

%%%%%%%

\section{Conclusion and Discussion}\label{sec:conclusion}
We have presented a linear time algorithm to solve the optimal $k$-step function problem,
when $k$ a constant.
Most of the effort is spent on identifying a ``big partition.''
It is desirable to reduce the constant of proportionality. 

The {\em size-$k$ histogram construction problem}~\cite{guha2007},
where the points are not weighted, 
is similar to the problem we addressed in this paper.
Its generalized version,
where the points are weighted,
is equivalent to our problem, and thus can be solved in optimal linear time when $k$ is a constant.
The {\em line-constrained $k$ center problem} is defined by:
Given a set $P$ of weighted points in the plane and a horizontal line $L$,
determine $k$ centers on $L$
such that the maximum weighted distance of the points to their closest centers is minimized.
This problem was solved in optimal $O(n\log n)$ time for arbitrary $k$ even if the points
are sorted~\cite{karmakar2013,wang2014a}.
Our algorithm presented here can be applied to solve this problem 
in $O(n)$ time if $k$ is a constant.

A possible extension of our work reported here is to use a cost other than the weighted vertical distance.
There is a nice discussion in \cite{guha2007} on the various measures one can use. 
Our complexity results are valid
if the cost is more general than (\ref{eqn:pointCost}),
%is given by $D(p, F(x))\triangleq f(d(p, F(x))) w(p)$,
%Here, $f(\cdot)$ is a monotonically increasing convex or concave function,
in particular, $D(p, F(x))\triangleq d(p, F(x))^2 w(p)$,
which is often used as an error measure.
%%%%
\section*{Acknowledgement}\label{sec:ack}
This work was supported in part by Discovery Grant \#13883 from
the Natural Science and Engineering Research Council (NSERC) of Canada and in part by MITACS,
both awarded to Bhattacharya.

%\bibliographystyle{../../splncs03}
%\bibliography{../../facilityLocation/refs}

\end{document}